\documentclass[11pt]{amsart}

\newcommand{\HH}{\mathcal{H}}

\newcommand{\RR}{\mathbb{R}}
\newcommand{\CC}{\mathbb{C}}

\newcommand{\LL}{\mathcal{L}}

\newcommand{\dprime}{{\prime\prime}}

\newcommand{\Inn}[2]{\langle #1, #2 \rangle}

\newcommand{\paulix}{\bigl(\begin{smallmatrix} & 1 \\ 1 & \end{smallmatrix}\bigr)}
\newcommand{\pauliy}{\bigl(\begin{smallmatrix} & -i\\ i & \end{smallmatrix}\bigr)}
\newcommand{\pauliz}{\bigl(\begin{smallmatrix}1 & \\ & -1 \end{smallmatrix}\bigr)}

\makeatletter
\let\uppercasenonmath\@gobble

\newtheorem{prop}{Proposition}

\title{Kramers degeneracy without eigenvectors}
\author[Bryan W. Roberts]{Bryan W. Roberts\\University of Southern California\\bryan.roberts@usc.edu}

\date{\today}

\begin{document}
	
\maketitle

\begin{abstract}
	Wigner gave a well-known proof of Kramers degeneracy, for time reversal invariant systems containing an odd number of half-integer spin particles. But Wigner's proof relies on the assumption that the Hamiltonian has an eigenvector, and thus does not apply to many quantum systems of physical interest. This note illustrates an algebraic way to talk about Kramers degeneracy that does not appeal to eigenvectors, and provides a derivation of Kramers degeneracy in this more general context.
\end{abstract}

\section{Introduction}

Wigner's derivation of Kramers degeneracy \cite{wigner1932kramer}, for time reversal invariant systems with an odd number of half-integer spin particles, has become standard repertoire in recent studies of the phenomenon \cite{bardarson2008a,KravtsovZirnbauer1992a,Liu2011a,Loss2009a,Meidan2011a,Ramazashvili2009a}. But Wigner's derivation assumes that the self-adjoint operator generating the dynamics (henceforth, the Hamiltonian) has at least one eigenvector. Mathematicians have long known that there are Hamiltonians of physical interest that do not have any eigenvectors. A closed linear operator $H$ has an eigenvector if the subspace of non-zero vectors $\psi$ such that $(H-\lambda)\psi=\mathbf{0}$ for some constant $\lambda$ is not empty. When the entire spectrum of $H$ can be expressed as values of $\lambda$ in such an equation, $H$ is said to have a \emph{pure point spectrum}. But the subspace can also be empty, in which case $H$ has no eigenvectors, and is said to have a \emph{purely continuous spectrum}. For example, in the Hilbert space $\LL^2(\RR)$ of square integrable functions $\psi(x):\RR\rightarrow\CC$, the multiplication operator $Q\psi(x):=x\psi(x)$ has no eigenvectors. Neither does the operator $H=P^2+Q$ in the same space, where $P:=-id/dx$.

This note illustrates a way to talk about Kramers degeneracy without eigenvectors, and provides a derivation of Kramers degeneracy in this more general context. The next section recapitulates Wigner's derivation of Kramers degeneracy. Section \ref{sec:continuous-degeneracy} introduces an algebraic generalization of degeneracy, which applies to all self-adjoint operators regardless of whether or not they admit an eigenvector. Section \ref{sec:continuous-kramers} then provides a derivation of Kramers degeneracy in this more general context. Like Wigner's derivation, our proof relies the premise that time reversal satisfies $T^2=-I$. For completeness, the scope of that premise is discussed in Section \ref{sec:discussion}.

\section{Wigner's proof of Kramers degeneracy}\label{sec:wigners-proof}

A self-adjoint operator with a pure point spectrum is degenerate if it has two orthogonal eigenvectors with the same eigenvalue. Wigner \cite{wigner1932kramer} gave the following proof that this happens whenever a fermion, or any other system in which the time reversal operator satisfies $T^2=-1$, happens to be time reversal invariant.

\begin{prop}[Wigner]\label{prop:wigner}
Let $\HH$ be a Hilbert space. Let $H$ be a linear self-adjoint operator with a non-empty point spectrum. Let $T:\HH\rightarrow\HH$ be an antiunitary bijection. If $T^2=-I$ (fermion condition) and $[T,H]=0$ ($T$-invariance), then $H$ has two orthogonal eigenvectors with the same eigenvalue (degeneracy).
\end{prop}
\begin{proof}
Since $H$ has a non-empty point spectrum, there exists an eigenvector $\psi$ with a real eigenvalue $\lambda$ such that $H\psi=\lambda\psi$. Let $\psi^\prime=T\psi$. By $T$-invariance, $H\psi^\prime = HT\psi=TH\psi = \lambda T\psi = \lambda\psi^\prime$. The eigenvectors $\psi$ and $\psi^\prime$ thus share the same eigenvalue $\lambda$. They are moreover orthogonal, since
\begin{equation*}
	\Inn{\psi}{T\psi} 	= \Inn{T\psi}{T^2\psi}^*
						= -\Inn{T\psi}{\psi}^*
						= -\Inn{\psi}{T\psi},
\end{equation*}
where the first equality follows from the fact that $T$ is antiunitary, and the second from the fermion condition. Therefore, $\Inn{\psi}{T\psi}=0$.
\end{proof}

Note that both Wigner's characterization of degeneracy, as well as his proof, rely on the existence of an eigenvector for the Hamiltonian.

\section{Algebraic definition of degeneracy}\label{sec:continuous-degeneracy}

There is a natural generalization of degeneracy to operators that do not have a pure point spectrum, which is algebraic in character. If an operator $A$ is degenerate in the sense of having two orthogonal eigenvectors with the same eigenvalue $\lambda$, then the set of vectors $\psi$ for which
\begin{equation}\label{eq:multiplicity}
	(A - \lambda I)\psi = 0
\end{equation}
forms a subspace of dimension at least 2. More generally, the dimension of this subspace will be equal to the number of orthogonal eigenvectors satisfying Equation \eqref{eq:multiplicity}, called the \emph{multiplicity} of $\lambda$. So, $A$ is non-degenerate if and only if all its eigenvalues have multiplicity $1$.

To generalize this concept, let $\{A\}^\prime$ be the \emph{commutant} of $A$, meaning the set of bounded linear operators that commute with $A$. Let $\{A\}^\dprime$ be the \emph{double-commutant} of $A$, meaning the set of closed linear operators that commute with all the elements of $\{A\}^\prime$. These sets are related to the previous notion of degeneracy by the following fact.

\begin{prop}[Blank, Exner, Havl\'i\v{c}ek]
	Let $A$ be a self-adjoint linear operator with a pure point spectrum. Then all the eigenvectors of $A$ have multiplicity 1 if and only if $\{A\}^\prime = \{A\}^\dprime$.
\end{prop}
\begin{proof}
	\cite[Theorem 5.8.6]{BlankExnerHavlicek}
\end{proof}

For self-adjoint operators with a pure point spectrum, non-degeneracy is thus equivalent to the condition that $\{A\}^\prime = \{A\}^\dprime$. But the latter condition, unlike the multiplicity of an eigenvalue of $A$, is well-defined even when $A$ has no eigenvectors.

\section{Kramers degeneracy without eigenvectors}\label{sec:continuous-kramers}

The statement $\{A\}^\prime = \{A\}^\dprime$ generalizes the usual notion of non-degeneracy; thus, the statement $\{A\}^\prime \neq \{A\}^\dprime$ generalizes the notion of degeneracy. This allows for the following more general derivation of Kramers degeneracy for time reversal invariant systems with an odd number of half-integer spin particles.

\begin{prop}\label{prop:continuous-wigner}
Let $H$ be a (possibly continuous-spectrum) self-adjoint operator on a separable Hilbert space, which is not the zero operator. Let $T:\HH\rightarrow\HH$ be an antiunitary bijection. If $T^2=-I$ (fermion condition) and $[T,H]=0$ ($T$-invariance), then $\{H\}^\prime \neq \{H\}^\dprime$ (general degeneracy).
\end{prop}
\begin{proof}
	We show that the supposition $\{H\}^\prime=\{H\}^\dprime$ leads to a contradiction. Let $K$ be the conjugation operator in the $H$ basis\footnote{The \emph{conjugation operator in the $H$ basis} is defined as follows.  Let $\psi(x)$ be a square integrable function in the spectral representation of $H$. Then $K$ is the unique bijection such that $T\psi(x)=\psi^*(x)$ for all $\psi(x)$. It follows that $K$ is antiunitary, $K^2=I$, and $[K,H]=0$. See \cite[\S XV.6]{messiah1999} for an introduction.}. Since $T$ is antiunitary, there exists a unitary operator $U$ such that $T=UK$. By $T$-invariance,
\begin{equation*}
	0=[T,H]=(UK)H - H(UK) = (UH)K - (HU)K,
\end{equation*}
where the last line applies the fact that $[K,H]=0$. This implies $UH - HU = 0$. So, $U$ is a bounded linear operator that commutes with $H$, and hence $U\in\{H\}^\prime$. Our hypothesis $\{H\}^\prime=\{H\}^\dprime$ now implies that $U\in\{H\}^\dprime$. By von Neumann's double-commutant theorem, it follows that we can write $U=f(H)$ for some function in the weak closure of $H$. But $U$ is unitary, and so can be expressed as $U=e^{iS}$ for some self-adjoint operator $S$ \cite[Proposition 5.3.8]{BlankExnerHavlicek}. Combining these two facts, we have that $U=e^{ig(H)}$, where $g(H)=S$ is self-adjoint, and is thus a real-valued function. But for real-valued $g$, $[g(H),K]=0$, and so $KUK^{-1}=e^{Kig(H)K} = e^{-ig(H)K^2}=U^*$ (where we have applied the fact that $K^2=I$ in the last equality). Therefore,
\begin{equation*}
	T^2 = UKUK = UU^* = I.
\end{equation*}
This contradicts the fermion condition.
\end{proof}

\section{Discussion of the $T^2=-I$ premise}\label{sec:discussion}

The assumption that $T^2=-I$ is a generic feature of any system satisfying the Pauli spin relations for an odd number of particles. Here we will illustrate the derivation for a single particle system; the generalization to more complex systems is straightforward.

Let $\sigma_1 = \paulix$, $\sigma_2 = \pauliy$, $\sigma_3 = \pauliz$ be the generators of the (irreducible) Pauli representation. The time reversal operator $T$ for such a system may be assumed to have the following properties.
	\begin{enumerate}
		\item \emph{$T$ is antiunitary}. This is a feature of time reversal in any context, first pointed out by Wigner \cite{wigner1931}.
		\item \emph{$T$ is an involution, $T^2=e^{i\theta}I$.} This captures what it means for $T$ to be a ``reversal'': apply it twice, and you get back to where you started up to a phase factor $e^{i\theta}$.
		\item \emph{$T$ reverses angular momentum.} Since each $\sigma_i$ represents a kind of angular momentum, they each reverse sign under time reversal. This can be viewed as stemming from the fact that the $\sigma_i$ are generators of a rotation group $R_\theta=e^{i\theta\cdot\sigma}$. Insofar as time reversal does not pick out any preferred direction with respect to this group, one may assume that $T$ commutes with $R_\theta$. But $T$ is antiunitary, and if an antiunitary operator $T$ commutes with $R_{\theta}=e^{i\theta\cdot\sigma}$, then it must \emph{anti}commute with the generator $\sigma$.
	\end{enumerate}
These properties of the time reversal operator imply the following.

	\begin{prop}
		Let $\sigma_1$, $\sigma_2$, $\sigma_3$ be the spin operators in the Pauli representation, and let $K$ be the conjugation operator in the $\{\binom{1}{0},\binom{0}{1}\}$ basis. If $T$ is any antilinear involution that satisfies $T\sigma_iT^{-1}=-\sigma_i$ for each $i=1,2,3$, then $T=k\sigma_2K$ for some complex unit $k$, and $T^2=-I$.	
	\end{prop}
	\begin{proof}
	 	Define $\tilde{T}:=\sigma_2 K$, and let $T$ be any involution that reverses the sign of $\sigma_i$ for each $i=1,2,3$. Then,
	\begin{equation*}
	   (\tilde{T}T)\sigma_i(\tilde{T}T)^{-1} = \tilde{T}(T\sigma_iT^{-1})\tilde{T}^{-1} = \tilde{T}(-\sigma_i)\tilde{T}^{-1} = \sigma_i.
	\end{equation*}
	$\tilde{T}T$ commutes with all the generators of the representation, and so it commutes with everything. But the Pauli representation is irreducible, so $\tilde{T}T=-kI$ for some $k\in\CC$ by Schur's lemma \cite[Theorem 6.7.1]{BlankExnerHavlicek}. Multiplying on the left by $-\tilde{T}$ and recalling that $\tilde{T}^2=-I$,
	\begin{equation*}
		T = k\tilde{T}=k\sigma_2K.
	\end{equation*}
	This $T$ is an involution, so there is a $c\in\CC_{unit}$ such that $cI = T^2 = (k\tilde{T})(k\tilde{T}) = kk^*\tilde{T}^2 = -kk^*$. But $-kk^*$ is real and negative for any complex constant $k$, and the only complex unit that is real and negative is $c=-1$. Thus, $c=-1$, and so $kk^*=-c=1$. This implies that $k$ is a complex unit and $T^2=-I$.
	\end{proof}  
Note that the particular choice of representation adopted here is immaterial, since such irreducible representations of the Pauli relations are unitarily equivalent.



\end{document}